\documentclass{llncs}
 
\usepackage{amssymb}
\usepackage{amsmath}
\usepackage{bm,wasysym}
\usepackage{xcolor}
\usepackage{pgf}

\begin{document}

\title{Existence of Secure Equilibrium\\ in Multi-Player Games with Perfect Information} 

\author{Julie De Pril\inst{1} \and J\'{a}nos Flesch\inst{2} \and Jeroen Kuipers\inst{3} \and \\Gijs Schoenmakers\inst{3} \and Koos Vrieze\inst{3}}

\institute{D\'epartement de Math\'ematique, Universit\'e de Mons, Place du Parc, 20, 7000 Mons, Belgium. \and Department of Quantitative Economics, Maastricht University, P.O. Box 616, 6200~MD, Maastricht, The Netherlands. \and Department of Knowledge Engineering, Maastricht University, P.O. Box 616, 6200~MD, Maastricht, The Netherlands.}

\maketitle

\begin{abstract}
Secure equilibrium is a refinement of Nash equilibrium, which provides some security to the players against deviations when a player changes his strategy to another best response strategy. The concept of secure equilibrium is specifically developed for assume-guarantee synthesis and has already been applied in this context. Yet, not much is known about its existence in games with more than two players. In this paper, we establish the existence of secure equilibrium in two classes of multi-player perfect information turn-based games: (1)~in games with possibly probabilistic transitions, having countable state and finite action spaces and bounded and continuous payoff functions, and (2)~in games with only deterministic transitions, having arbitrary state and action spaces and Borel payoff functions with a finite range (in particular, qualitative Borel payoff functions). We show that these results apply to several types of games studied in the literature.
\end{abstract}

\section{Introduction}

\noindent\textbf{The game:} We examine multi-player perfect information turn-based games with possibly probabilistic transitions. In such a game, each state is associated with a player, who controls this state. Play of the game starts at the initial state. At every state that play visits, the player who controls this state has to choose an action from a given action space. Next, play moves to a new state according to a probability measure, which may depend on the current state and the chosen action. This induces an infinite sequence of states and actions, and depending on this play, each player receives a payoff.

These payoffs are fairly general. For example, they could arise as some aggregation of instantaneous rewards that the players receive at the periods of the game. A frequently used aggregation would be taking the discounted sum of the instantaneous rewards. Another special case arises if the payoffs only depend on the first $T$ periods, which essentially results in a game of horizon $T$. Or as another example, the payoffs could represent reachability conditions, and then a player's payoff would be either 1 or 0 depending on whether a certain set of states is reached.

Two-player zero-sum games with possibly probabilistic transitions have been applied in the model-checking of reactive systems where randomness occurs, because they allow to model the interactions between a system and its environment. However, complex systems are usually made up of several components with objectives that are not necessarily antagonist, that is why multi-player non zero-sum games are better suited in such cases.
\smallskip

\noindent\textbf{Nash equilibrium:} In these games, Nash equilibrium is a prominent solution concept. A Nash equilibrium is a strategy profile such that no player can improve his payoff by individually deviating to another strategy. In various classes of perfect information games, a Nash equilibrium is known to exist (cf. for example \cite{FL83,Har85} or the result of Mertens and Neyman in \cite{Mer87}).\smallskip

\noindent\textbf{Secure equilibrium:} Despite the obvious appeal of Nash equilibrium, certain applications call for additional properties. Chatterjee et al \cite{CHJ06} introduced the concept of secure equilibrium, which they specifically designed for assume-guarantee synthesis. They gave a definition of secure equilibrium (\cite[Definition~8]{CHJ06}) in qualitative $n$-player games\footnote{Note that Chatterjee et al \cite{CHJ06} gave no existence result in the $n$-player case.}, and then a characterization (\cite[Proposition~4]{CHJ06}), which however turns out not to be equivalent. With their definition, such kind of equilibirum may fail to exist even in very simple games (cf. Remark~\ref{rem-sum-secure} or \cite[Example 2.2.34]{D13}). That is why we choose to call a \emph{strongly secure equilibrium} an equilibrium according to \cite[Definition~8]{CHJ06} (see Remark~\ref{rem-sum-secure}), and we choose to call, as it has already been done in \cite{BBDG12,BBDG13},  a \emph{secure equilibrium} an equilibrium according to the alternative characterization given in \cite[Proposition~4]{CHJ06}, extended to the quantitative framework. Note that, with these definitions, every secure equilibrium is automatically strongly secure if there are only two players.

Thus, a strategy profile is called a secure equilibrium if it is a Nash equilibrium and moreover the following security property holds: if any player individually deviates to another best response strategy, then it cannot be the case that all opponents are weakly worse off due to this deviation and at least one opponent is even strictly hurt. For applications of secure equilibrium, we refer to \cite{CHJ06,CH07,CR10}. 

Only little is known about the existence of secure equilibrium. To our knowledge, the available existence results are for games with only two players\footnote{However, Chatterjee and Henzinger \cite{CH07} proved the existence of secure equilibria in the special case of 3-player qualitative games where the third player can win unconditionally.} and for only deterministic transitions, i.e. every action in every state leads to a certain state with probability 1. Chatterjee et al \cite{CHJ06} proved the existence of a secure equilibrium in two-player games in which the payoff function of each player is the indicator function of a Borel subset of plays. Recently,  the existence of a secure equilibrium, even a subgame-perfect secure equilibrium, has been shown \cite{BBDG13} in two-player games in which each player's goal is to reach a certain set of states as quickly as possible and his payoff is determined by the number of moves it takes to get there. Very recently, the existence of a secure equilibrium has been proved \cite{BMR14} in a class of two-player quantitative games which includes payoff functions like $\sup$, $\inf$, $\limsup$, $\liminf$, mean-payoff, and discounted sum.\smallskip

\noindent\textbf{Our contribution:} In this paper, we address the existence problem of secure equilibrium for multi-player perfect information games. We establish the existence of secure equilibrium in two classes of such games. First, when probabilistic transitions are allowed\footnote{Such games are also called \emph{turn-based multi-player stochastic games}.}, we prove that a secure equilibrium exists, provided that the state space is countable, the action spaces are finite, and the payoff functions are bounded and continuous. To our knowledge, it is the first existence result of secure equilibria in \emph{multi-player quantitative} games. Moreover, our result extends to lexicographic preferences (see Section~\ref{sec:ccl}). Second, for games with only deterministic transitions, we prove that a secure equilibrium exists if the payoff functions are Borel measurable and have a finite range (in particular, for Borel qualitative objectives). For the latter result, we impose no restriction on the state and action spaces. We demonstrate that these results apply to several classes of games studied in the literature. Regarding proof techniques, the proof of the first result relies on an inductive procedure that eliminates certain actions of the game, while the proof of the second result exploits a transformation of the payoffs.\smallskip

\noindent\textbf{Structure of the paper:} Section~\ref{sec:model} is dedicated to the model. In Section~\ref{sec:results} we present the main results and mention some classes of games to which the results apply. Sections~\ref{sec:proof1} and~\ref{sec:proof2} contain the formal proofs of the results. Finally, Section~\ref{sec:ccl} concludes with some remarks and an algorithmic result for quantitative reachability objectives. 

\section{The Model}\label{sec:model}

We distinguish two types of perfect information games. The first type of games may contain transition probabilities between 0 and 1, whereas the second type only has deterministic transitions.

\subsection{Games with Probabilistic Transitions} A multi-player perfect information game with possibly probabilistic transitions is given by:
\begin{enumerate}
\item A finite set of players $N$, with $|N|\geq 2$. 
\item A countable state space $S$, containing an initial state $\widetilde{s}$.
\item A controlling player $i(s)\in N$ for every state $s\in S$.
\item A nonempty and finite action space $A(s)$ for every state $s\in S$. 
\item A probability measure $q(s,a)$ for every state $s\in S$ and action $a\in A(s)$, which assigns, to every $z\in S(s,a)$ (the set of possible successor states when choosing action $a$ in state $s$), the probability $q(s,a)(z)$ of transition from state $s$ to state $z$ under action $a$.

Let $\mathbb{N}=\{0,1,2,\ldots\}$. Let $\mathcal{H}$ be the set of all sequences of the form $(s^0,a^0,\ldots,$\linebreak$s^{n-1},a^{n-1},s^n)$, where $n\in\mathbb{N}$, such that $s^0=\widetilde{s}$, and $a^m\in A(s^m)$ and $s^{m+1}\in S(s^m,a^m)$ for all $m=0,1,\ldots,n-1$. Let $\mathcal{P}$ be the set of all infinite sequences of the form $(s^m,a^m)_{m\in \mathbb{N}}$ such that $s^0=\widetilde{s}$, and $a^m\in A(s^m)$ and $s^{m+1}\in S(s^m,a^m)$ for all $m\in\mathbb{N}$. The elements of $\mathcal{H}$ are called histories and the elements of $\mathcal{P}$ are called plays. A history $h$ is called a prefix of a play $p$, denoted by $h\prec p$, if $p$ starts with $h$. We endow $\mathcal{P}$ with the topology induced by the cylinder sets $\mathcal{C}(h) = \{p\in\mathcal{P}|h\prec p\}$ for $h\in \mathcal{H}$. In this topology, a sequence of plays $(p_m)_{m\in\mathbb{N}}$ converges to a play $p$ precisely when for every $k\in\mathbb{N}$ there exists an $N_k\in\mathbb{N}$ such that $p_m$ coincides with $p$ on the first $k$ coordinates for every $m\geq N_k$.

\item A payoff function $u_i:\mathcal{P}\rightarrow \mathbb{R}$ for every player $i\in N$, which is bounded and Borel measurable.
\end{enumerate}

The game is played as follows at periods in $\mathbb{N}=\{0,1,2,\ldots\}$. Play starts at period 0 in state $s^0=\widetilde{s}$, where the controlling player $i(s^0)$ chooses an action $a^0$ from $A(s^0)$. Then, transition occurs according to the probability measure $q(s^0,a^0)$ to a state $s^1$. At period 1, the controlling player $i(s^1)$ chooses an action $a^1$ from $A(s^1)$. Then, transition occurs according to the probability measure $q(s^1,a^1)$ to a state $s^2$, and so on. The realization of this process is a play $p=(s^0,a^0,s^1,a^1,s^2,\ldots)$, and each player $i\in N$ receives payoff $u_i(p)$.

We can assume without loss of generality that the sets $S(s,a)$, for $s\in S$ and $a\in A(s)$, are mutually disjoint, and their union is $S$. This means that each state can be visited in exactly one way from the initial state $\widetilde{s}$, so there is a bijection between states and histories. For this reason, we will work with states instead of histories. \smallskip

\noindent\textbf{Strategies:} 
For every player $i\in N$, let $S_i$ denote the set of those states (histories) which are controlled by him. A strategy for a player $i\in N$ is a function $\sigma_i$ that assigns an action $\sigma_i(s)\in A(s)$ to every state $s\in S_i$. The interpretation is that $\sigma_i(s)$ is the recommended action if state $s$ is reached. A strategy profile is a tuple $(\sigma_1,\ldots,\sigma_{|N|})$ where $\sigma_i$ is a strategy for every player $i\in N$. Given a strategy profile $\sigma=(\sigma_1,\ldots,\sigma_{|N|})$ and a player $i$, we denote by $\sigma_{-i}$ the profile of strategies of player $i$'s opponents, i.e. $\sigma_{-i}=(\sigma_j)_{j\in N,\, j\neq i}$. 

A strategy profile $\sigma$ induces a unique probability measure on the sigma-algebra of the Borel sets of $\mathcal{P}$. The corresponding expected payoff for player $i\in N$ is denoted by $u_i(\sigma)$.\smallskip

\noindent\textbf{Nash equilibrium:} A strategy profile $\sigma^\ast$ is called a Nash equilibrium, if no player can improve his expected payoff by a unilateral deviation, i.e. $u_i(\tau_i,\sigma^\ast_{-i})\leq u_i(\sigma^\ast)$ for every player $i\in N$ and every strategy $\tau_i$ for player $i$. In other words, every player plays a best response to the strategies of his opponents. \smallskip

\noindent\textbf{Secure equilibrium:}
A strategy profile $\sigma^\ast$ is called a secure equilibrium, if it is a Nash equilibrium and if, additionally, no player $i\in N$ has a strategy $\tau_i$ with $u_i(\tau_i,\sigma^\ast_{-i})=u_i(\sigma^\ast)$ such that we have for all players $j\in N\setminus\{i\}$ that $u_j(\tau_i,\sigma^\ast_{-i}) \leq u_j(\sigma^\ast)$ and for some player $k\in N\setminus\{i\}$ that $u_k(\tau_i,\sigma^\ast_{-i})<u_k(\sigma^\ast)$.

The interpretation of the additional property is the following. Consider a Nash equilibrium $\sigma^\ast$ and a strategy $\tau_i$ for some player $i$. By deviating to $\tau_i$, player $i$ either receives a worse expected payoff than with his original strategy $\sigma_i^\ast$, or the same expected payoff at best. In the former case, it is not in player $i$'s interest to deviate to $\tau_i$. In the latter case, however, even though player $i$ is indifferent, his opponents could get hurt. The property thus prevents the case that this deviation is weakly worse for all opponents of player $i$, and that this deviation is even hurting a player. So, in a certain sense, player $i$´s opponents are secure against such deviations by player $i$.

An equivalent formulation of secure equilibrium is the following: a strategy profile $\sigma^\ast$ is called a secure equilibrium, if it is a Nash equilibrium and if, additionally, the following property holds for every player $i\in N$: if $\tau_i$ is a strategy for player $i$ such that $u_i(\tau_i,\sigma^\ast_{-i})=u_i(\sigma^\ast)$ and $u_j(\tau_i,\sigma^\ast_{-i}) < u_j(\sigma^\ast)$ for some player $j\in N$, then there is a player $k\in N$ such that $u_k(\tau_i,\sigma^\ast_{-i})>u_k(\sigma^\ast)$.

\begin{remark}\label{rem-sum-secure} There is a refinement of secure equilibrium, which plays an important role in the proofs of our main results, Theorems~\ref{main-th1} and~\ref{main-th2}. We call a strategy profile $\sigma^\ast$ a \emph{sum-secure equilibrium}, if it is a Nash equilibrium and if additionally the following property holds for every player $i\in N$: if $\tau_i$ is a strategy for player $i$ such that $u_i(\tau_i,\sigma^\ast_{-i})=u_i(\sigma^\ast)$ then $\sum_{j\in N\setminus\{i\}}u_j(\tau_i,\sigma^\ast_{-i})\geq \sum_{j\in N\setminus\{i\}}u_j(\sigma^\ast)$. In fact, in Theorems~\ref{main-th1} and~\ref{main-th2}, we prove the existence of sum-secure equilibria.

An even stronger concept is the following. A strategy profile $\sigma^\ast$ is called a \emph{strongly secure equilibrium}, if it is a Nash equilibrium and if additionally the following property holds for every player $i\in N$: if $\tau_i$ is a strategy for player $i$ such that $u_i(\tau_i,\sigma^\ast_{-i})=u_i(\sigma^\ast)$ then $u_j(\tau_i,\sigma^\ast_{-i})\geq u_j(\sigma^\ast)$ for all players $j\neq i$. Note that every secure equilibrium is also strongly secure if there are only two players. The concept of strongly secure equilibrium has the serious drawback though that it fails to exist even in very simple games \cite[Example 2.2.34]{D13}. Indeed, consider the following game with three players. Player 1 has two actions, such that the first action yields payoffs $(1,2,0)$, whereas the second action yields payoffs $(1,0,2)$. Players 2 and 3 never become active in this game. Now either action of player 1 yields a secure equilibrium, but neither of them is strongly secure.
\end{remark}

\subsection{Games with Deterministic Transitions} Another type of perfect information games arises when the game has only deterministic transitions, i.e. when the set $S(s,a)$ is a singleton for every state $s\in S$ and every action $a\in A(s)$. In this case, we do not need to take care of measurability conditions for the calculation of expected payoffs. Hence, we can drop the assumptions that the state space is countable and the action spaces are finite, and they can be arbitrary. 

\section{The Main Results}\label{sec:results}

In this section, we present and discuss our main results for the existence of secure equilibrium. First we examine the case of probabilistic transitions.

\begin{theorem}\label{main-th1} Take a perfect information game, possibly having probabilistic transitions, with countable state and finite action spaces. If every player's payoff function is bounded and continuous\footnote{Notice that any continuous function is Borel measurable.}, then the game admits a secure equilibrium.
\end{theorem}

Let us first elaborate on the conditions of Theorem~\ref{main-th1}. By assuming that the state space is countable, we avoid measure theoretic complications. Without the assumptions that the action spaces are finite and the payoff functions are continuous, even a Nash equilibrium may fail to exist. This is shown by the following two examples.

\begin{example}\label{choosing-integer} Consider the following one shot game (each finite duration game can be seen as an infinite duration game): a player can choose any positive integer, and if he chooses $k$ then his payoff is $1-\tfrac{1}{k}$. The payoff function is continuous, but the action space is infinite. Clearly, the game admits no Nash equilibrium. 
\end{example}

\begin{example}
Suppose that a player can decide every day whether to stop or to continue. If he stops at period $n$, then his payoff is $1-\tfrac{1}{n}$, whereas if he never stops then he receives payoff 0. Here, the action space is finite, but the payoff function is not continuous. This game has no Nash equilibrium.
\end{example}

For games with only deterministic transitions, we also obtain the following result, which imposes a condition on the range of the payoff functions.

\begin{theorem}\label{main-th2} 
Take a perfect information game with deterministic transitions, with arbitrary state and action spaces. If every player's payoff function is Borel measurable and has a finite range, then the game admits a secure equilibrium.
\end{theorem}

Theorem~\ref{main-th2} assumes that the range of the payoff functions is finite. This assumption is only useful for deterministic transitions, otherwise the range of the expected payoffs may become infinite. Without this assumption, even a Nash equilibrium can fail to exist, as Example~\ref{choosing-integer} above shows. 

The methods of proving Theorems~\ref{main-th1} and~\ref{main-th2} are really different. The proof of Theorem~\ref{main-th1} uses an inductive procedure that eliminates certain actions in certain states. This procedure terminates with a game, in which one can identify an interesting strategy profile that can be enhanced with punishment strategies to be a secure equilibrium in the original game. The proof of Theorem~\ref{main-th2} relies on a transformation of the payoffs. In the new game, a Nash equilibrium exists, and it is a secure equilibrium of the original game.

The above two theorems apply to various classes of games that have been studied in the literature. We now mention a number of them. We only discuss the assumptions imposed on the payoff functions, as it is clear when a game satisfies the rest of the assumptions:
\begin{enumerate}
\item In \emph{games with a finite horizon}, the payoff functions are continuous, so Theorem~\ref{main-th1} is applicable to such games if the other hypotheses are satisfied. The same observation holds for \emph{discounted games}, where the players aggregate instantaneous payoffs by taking the discounted sum.\smallskip

\item All qualitative payoff functions have a finite range, so Theorem~\ref{main-th2} directly implies the existence of a secure equilibrium in multi-player games played on a graph with deterministic transitions and \emph{qualitative Borel objectives} (in particular, objectives like reachability, safety, (co--)B\"{u}chi, parity,\ldots). \smallskip

\item Now consider a game played on a graph with quantitative payoff functions.  For \emph{quantitative reachability objectives}, by \cite[Remark~2.5]{BBDG13}, we can use Theorem~\ref{main-th1} to find a secure equilibrium as long as the transitions are deterministic (the transformation drastically changes expected payoffs). For \emph{quantitative safety objectives}, continuous payoff functions can also be defined: if $F_i\subseteq S$ is the safety set of player~$i$, let $u_i(p)=1-\tfrac{1}{n+1}$  for any play $p$ if, along $p$, the set $S\setminus F_i$ is reached at period $n$ for the first time and let $u_i(p)=1$ if $S\setminus F_i$ is never reached along $p$. And in a similar way for \emph{quantitative B\"{u}chi objectives}: if $B_i \subseteq S$ is the B\"uchi set of player~$i$, let $M(p)$ denote, for any play $p$, the set of periods at which the play $p$ is in a state in $B_i$. Then, define $u_i(p)=\sum_{k\in M(p)} \tfrac{1}{2^k}$.\smallskip

\item Theorem~\ref{main-th2} implies the existence of a secure equilibrium in games played on \emph{finite weighted\footnote{Each edge of the graph is labelled by a $|N|$-tuple of real values.} graphs} with deterministic transitions, where the payoff functions are computed as the $\sup$, $\inf$, $\limsup$ or $\liminf$ of the weights appearing along plays (as in \cite{BMR14}). Indeed, these functions have finite range in such games.
\end{enumerate}

\section{The proof of Theorem \ref{main-th1}}\label{sec:proof1}

In this section, we provide a formal proof of Theorem~\ref{main-th1}. Consider a perfect information game $G$, having possibly probabilistic transitions, with a countable state and finite action spaces. Assume that the payoff function of every player is bounded and continuous. 

\subsection{Preliminaries}
In this subsection, we keep the game $G$ fixed. We introduce some notation, define some preliminary notions and state some properties of them.

For every player $i\in N$, we use the notation $\Sigma_i$ for the set of strategies of player $i$, so $\Sigma_i = \times_{s\in S_i} A(s)$. We endow $\Sigma_i$ with the product topology $\mathcal{T}_i$. Since the set $S_i$ is countable and the action spaces are finite, the topological space $(\Sigma_i,\mathcal{T}_i)$ is compact and metrizable. Hence, the set of strategy profiles $\Sigma=\times_{i\in N}\Sigma_i$, endowed with the product topology $\mathcal{T}=\times_{i\in N}\mathcal{T}_i$, is also compact and metrizable. So, these spaces are sequentially compact, meaning that every sequence in them has a convergent subsequence. This is one of the main consequences of assuming that the action spaces are finite. 

We assumed that the payoffs are continuous on the set of plays. The next lemma states that the expected payoffs are continuous as well. 

\begin{lemma}\label{cont-exp}
For any player $i\in N$, the expected payoff function $u_i:\Sigma\rightarrow\mathbb{R}$ is continuous.
\end{lemma}

\begin{proof}
Since $\Sigma$, endowed with the topology $\mathcal{T}$, is metrizable, it suffices to show that if a sequence $(\sigma^n)_{n\in\mathbb{N}}$ of strategy profiles converges to a strategy profile $\sigma$, then $u_i(\sigma^n)$ converges to $u_i(\sigma)$. So, take a sequence $(\sigma^n)_{n\in\mathbb{N}}$ of strategy profiles converging to some $\sigma$. Let $\mathbb{P}_{\sigma^n}$ denote the probability measure induced by $\sigma^n$ on the set of plays $\mathcal{P}$, and let $\mathbb{P}_{\sigma}$ be the one induced by $\sigma$. 

Consider an arbitrary cylinder set $\mathcal{C}(s)$ of $\mathcal{P}$. Because every state can be reached in exactly one way from the initial state $\widetilde{s}$, it follows that $\mathbb{P}_{\sigma^n}(\mathcal{C}(s))$ converges to $\mathbb{P}_{\sigma}(\mathcal{C}(s))$. Now take an open set $O$ of $\mathcal{P}$. Notice that $O$ can be written as a disjoint union of cylinder sets of $\mathcal{P}$. Since there are only countably many cylinder sets, we find that $\liminf_{n\rightarrow\infty} \mathbb{P}_{\sigma^n}(O) \geq \mathbb{P}_{\sigma}(O)$. This means that the probability measures $\mathbb{P}_{\sigma^n}$ converge weakly to $\mathbb{P}_{\sigma}$. Hence, because $u_i:\mathcal{P}\rightarrow\mathbb{R}$ is assumed to be bounded and continuous, $u_i(\sigma^n)=\mathbb{E}_{\sigma^n}(u_i)$ converges to $u_i(\sigma)=\mathbb{E}_{\sigma}(u_i)$, where $\mathbb{E}$ refers to the expectation.\qed
\end{proof}

Given a state $s\in S$, we define the subgame $G(s)$ as the game that arises when state $s$ is reached (i.e. past play has followed the unique history from the initial state $\widetilde{s}$ to $s$). Every strategy profile $\sigma$ induces a strategy profile in $G(s)$, as well as an expected payoff for every player $i$, which we denote by $u_{i}(\sigma|s)$. 

For every player $i\in N$, we derive a zero-sum perfect information game $G_i$ from $G$ by making the following modification. There are two players: player $i$ and an imaginary player $-i$, who replaces the set of opponents $N\setminus\{i\}$ of player $i$. So, whenever a state is reached controlled by a player from $N\setminus\{i\}$, player $-i$ can choose the action. Player $i$ tries to maximize his expected payoff given by $u_i$, whereas player $-i$ tries to minimize this expected payoff, so $u_{-i}=-u_i$. We say that this zero-sum game $G_i$ has a value, denoted by $v_i$, if player $i$ has a strategy $\sigma^\ast_i$ and his opponents have a strategy profile $\sigma^\ast_{-i}$, which is thus a strategy for player $-i$, such that $u_i(\sigma^\ast_i,\tau_{-i})\geq v_i$ for every strategy profile $\tau_{-i}$ and $u_i(\tau_i,\sigma^\ast_{-i})\leq v_i$ for every strategy $\tau_i$ for player $i$. This means that $\sigma_i^\ast$ guarantees for player $i$ that he receives an expected payoff of at least $v_i$, and $\sigma_{-i}^\ast$ guarantees that player $i$ does not receive an expected payoff of more than $v_i$. We call the strategy $\sigma^\ast_i$ optimal for player $i$ and the strategy profile $\sigma^\ast_{-i}$ optimal for player $i$'s opponents. Note that $(\sigma^*_i,\sigma^*_{-i})$ forms a Nash equilibrium in $G_i$, and vice versa, if $(\tau_i,\tau_{-i})$ is a Nash equilibrium in $G_i$, then $\tau_i$ and $\tau_{-i}$ are optimal. In a similar way, we can also speak of the value of the subgame $G_i(s)$ of $G_i$, for a state $s\in S$, which we denote by $v_i(s)$.

The next lemma states that, for every player $i\in N$ and state $s\in S$, the game $G_i(s)$ admits a value. Moreover, one can find a strategy for player $i$ that induces an optimal strategy in every game $G_i(s)$, for $s\in S$, and player $i$'s opponents have a similar strategy profile. This follows from the existence of a subgame-perfect Nash equilibrium in our setting (where subgame-perfect Nash equilibrium refers to a strategy profile that induces a Nash equilibrium in every subgame), see \cite[Theorem~1.2]{MS07}. It is also an easy extension of \cite[Corollary~4.2]{FL83}, with almost the same proof (based on approximations of the game by finite horizon truncations).

\begin{lemma}\label{val-exists}
Take a player $i\in N$. The value $v_i(s)$ of the game $G_i(s)$ exists for every state $s\in S$. Moreover, player $i$ has a strategy $\sigma^*_i$ such that $\sigma^*_i$ induces an optimal strategy in the subgame $G_i(s)$, for every state $s\in S$. Similarly, the opponents of player $i$ have a strategy profile $\sigma^*_{-i}$ such that $\sigma^*_{-i}$ induces an optimal strategy profile in the subgame $G_i(s)$, for every state $s\in S$.
\end{lemma}

For every player $i\in N$, every state $s\in S$ and every action $a\in A(s)$, define $$v_i(s,a)=\sum_{z\in S}\,q(s,a)(z)\cdot v_i(z).$$ This is in expectation the value for player $i$ in the subgame of $G$ that arises if player $i(s)$ chooses action $a$ in state $s$. Obviously, for the controlling player $i(s)$ we have 
\begin{equation}\label{contr-eq}
v_{i(s)} (s)=\max_{a\in A(s)} v_{i(s)}(s,a).
\end{equation} 
Here, the maximum is attained due to the finiteness of $A(s)$. Let us call an action $a\in A(s)$ optimal in state $s$ if $v_{i(s)}(s) = v_{i(s)}(s,a)$. We also have for every player $j\in N\setminus\{i(s)\}$ that 
\begin{equation}\label{noncontr-eq}
v_{j} (s)=\min_{a\in A(s)} v_{j}(s,a).
\end{equation}

\subsection{A Procedure to Derive a Restricted Game with Only Optimal Actions}

In this subsection, we define a procedure that inductively eliminates all actions that are not optimal and terminates with a specific game $G^\infty$, in which all actions are optimal.

Take a nonempty set $A'(s)\subseteq A(s)$ for every state $s\in S$. The sets $A'(s)$, for $s\in S$, induce a game $G'$ that we derive from $G$ as follows: the set of states $S'$ of $G'$ consists of those states $z\in S$ for which the unique history that starts at $\widetilde{s}$ and ends at $z$ only uses actions in the sets $A'(s)$, for $s\in S$. These are the states that a play can visit, with possibly probability zero, when the actions are restricted to the sets $A'(s)$, for $s\in S$. The action space of $G'$ in every state $s\in S'$ is then $A'(s)$. Further, the payoff functions of $G'$ are obtained by restricting the payoff functions of $G$ to plays corresponding to these new state and action spaces. 

Let $G^0$ be the game $G$, let $S^0=S$ and let $A^0(s)=A(s)$ for every $s\in S$. Then, at every state $s\in S$, we delete all actions that are not optimal. This results in a nonempty action space $A^1(s)\subseteq A^0(s)$ for every state $s\in S$. Let $G^1$ denote the induced game, with state space $S^1$. In the next step, at every state $s\in S^1$, we delete all actions from $A^1(s)$ that are not optimal in $G^1$. This gives a nonempty action space $A^2(s)\subseteq A^1(s)$ for every state $s\in S^1$. Let $G^2$ denote the induced game, with state space $S^2$. By proceeding this way, we obtain for every $k\in \mathbb{N}$ a game $G^k$ with state space $S^k$ and nonempty action spaces $A^k(s)$, for $s\in S^k$. Finally, let $S^\infty=\cap_{k\in\mathbb{N}} S^k$ and $A^\infty(s)=\cap_{k\in\mathbb{N}}A^k(s)$ for every $s\in S^\infty$. Note that the initial state $\widetilde{s}$ belongs to $S^\infty$, and also that, due to the finiteness of the action spaces, the sets $A^\infty(s)$, for $s\in S^\infty$, are nonempty. Let $G^\infty$ be the game induced by the sets $A^\infty(s)$, for $s\in S^\infty$. It is clear that the state space of $G^\infty$ is exactly $S^\infty$.

We now define a function $\phi:S\rightarrow \mathbb{N}\cup\{\infty\}$. Note that every state $s\in S$ belongs either to $S^k\setminus S^{k+1}$ for a unique $k\in\mathbb{N}$ or to $S^\infty$. In the former case, we define $\phi(s)=k$, whereas in the latter case we define $\phi(s)=\infty$. So, this is the latest iteration in which state $s$ is still included. For notational convenience, we extend the strategy space of player $i$ in $G^k$, for $k\in\mathbb{N}\cup\{\infty\}$, with strategies $\sigma_i$ for player $i$ in $G$ with the following property: for every $s\in S$, if $\phi(s)\geq k$ then $\sigma_i(s)\in A^k(s)$, whereas if $\phi(s)<k$ then $\sigma_i(s)$ is an arbitrary action in $A^{\phi(s)}(s)$. By doing so, every strategy in $G^k$ is also a strategy in $G^m$ if $m\leq k$.

Now we consider the subgames of the above defined games. For every $k\in\mathbb{N}\cup\{\infty\}$ and player $i\in N$, by Lemma~\ref{val-exists}, the game $G^k_i(s)$ has a value $v_i^k(s)$ for every $s\in S$. Moreover, player $i$ has a strategy $\sigma^k_i$ such that $\sigma^k_i$ induces an optimal strategy in the subgame $G^k_i(s)$ for every state $s\in S^k$. Similarly, the opponents of player $i$ have a strategy profile $\sigma^k_{-i}$ such that $\sigma^k_{-i}$ induces an optimal strategy profile in the subgame $G^k_i(s)$ for every state $s\in S^k$. Note that $\sigma^k_i$ can only use optimal actions in $G^k$, so by construction, $\sigma^k_i$ is also a strategy for player $i$ in the game $G^{k+1}$.

\begin{lemma} \label{inc_vk} We have $v_i^k(s)\leq v_i^{k+1}(s)$ for every $k\in\mathbb{N}$, player $i\in N$, and state $s\in S^{k+1}$.
\end{lemma}

\begin{proof} Take a $k\in\mathbb{N}$, a player $i\in N$, and a state $s\in S^{k+1}$. Let $\tau_{-i}$ be an arbitrary strategy profile for player $i$'s opponents in the game $G^{k+1}$. Hence, $\tau_{-i}$ is also a strategy profile for them in the game $G^{k}$. Therefore, by the optimality of $\sigma^k_{i}$ in $G^k_i(s)$, we have $u_i(\sigma^k_{i},\tau_{-i}|s) \geq v_i^k(s)$. Because  $\sigma^k_i$ is also a strategy for player $i$ in the game $G^{k+1}$, it means that player $i$ can guarantee at least $v_i^k(s)$ in the game $G^{k+1}(s)$. Therefore $v_i^k(s)\leq v_i^{k+1}(s)$, as desired.
\qed\end{proof}

\begin{lemma} \label{goodlimit} We have $\lim_{k\rightarrow\infty}v_i^k(s) = v_i^{\infty}(s)$ for every player $i\in N$ and state $s\in S^\infty$.
\end{lemma}

\begin{proof} Take a $k\in\mathbb{N}$, a player $i\in N$, and a state $s\in S^\infty$. By Lemma~\ref{inc_vk} and the boundedness of the payoffs, the sequence $(v_i^k(s))_{k\in\mathbb{N}}$ converges to some $r\in \mathbb{R}$. We need to show that $r=v_i^{\infty}(s)$.

Consider the sequences $(\sigma_{i}^k)_{k\in\mathbb{N}}$ and $(\sigma_{-i}^k)_{k\in\mathbb{N}}$. Since the strategy spaces in $G$ are sequentially compact, these sequences have convergent subsequences $(\sigma_{i}^k)_{k\in M}$ and $(\sigma_{-i}^k)_{k\in M}$, for some infinite subset $M\subseteq \mathbb{N}$. Let $\sigma_i$ and $\sigma_{-i}$ denote the limit strategies. It follows by construction that $\sigma_i$ is a strategy and respectively $\sigma_{-i}$ is a strategy profile in the game $G^\infty$.

Let $\tau_{-i}$ be an arbitrary strategy profile for
player $i$'s opponents in the game $G^\infty$. Then, $\tau_{-i}$ is also a strategy profile in the game $G^m$, for any $m\in M$. By the optimality of $\sigma_{i}^m$ in $G^m(s)$, we find $u_i(\sigma^m_{i},\tau_{-i}|s) \geq v_i^m(s)$ for all $m\in M$. If we take the limit when $m$ tends to infinity in the set $M$, Lemma~\ref{cont-exp} yields $u_i(\sigma_i,\tau_{-i}|s) \geq r$. Since $\sigma_i$ is a strategy in $G^\infty$, we have $v_i^\infty(s)\geq r$. It can be shown similarly, by using the strategy profile $\sigma_{-i}$, that $v_i^\infty(s)\leq r$. The proof is complete.
\qed\end{proof}

The next lemma shows that there is no need to continue the elimination procedure in a transfinite way.

\begin{lemma}\label{inf-game}
Every action is optimal in the game $G^\infty$, i.e. for every state $s\in S^\infty$ and action $a\in A^\infty(s)$, we have for the controlling player $i(s)$ that $v^\infty_{i(s)} (s)=v^\infty_{i(s)}(s,a).$ Moreover, if $\tau$ is a strategy profile in the game $G^\infty$, then $u_i(\tau|s)\geq v^\infty_i(s)$ for every player $i\in N$ and every state $s\in S^\infty$.
\end{lemma}

\begin{proof}

Take a state $s\in S^\infty$ and an action $a\in A^\infty(s)$. Then, $a\in A^k(s)$ for all $k\in \mathbb{N}$, and therefore $v^k_{i(s)}(s)=v_{i(s)}^k(s,a)$ for all $k\in \mathbb{N}$. By Lemma~\ref{goodlimit}, we obtain $v_{i(s)}^\infty(s)=v_{i(s)}^\infty(s,a)$, so the first part of the lemma follows.

Now take a strategy profile $\tau$ in the game $G^\infty$, a player $i$, and a state $s\in S^\infty$. Assume that state $s$ can be reached at period $n$ starting from the initial state $\widetilde{s}$. For $m>n$, let $\tau^m$ be the strategy profile that prescribes to play as follows: (1) at periods $0,1,\ldots,n-1$, arbitrary actions can be chosen in $G^\infty$, (2) at periods $n,n+1,\ldots,m-1$, the players follow $\tau$, (3) at periods $m,m+1,\ldots$, player $i$ follows the strategy $\sigma^\infty_i$ and all other players follow their strategies in $\tau_{-i}$. Let $s^\ell$ denote the random variable for the state at period $\ell\in\mathbb{N}$ according to $\tau^m$. By the first part of the lemma and by Equality~(\ref{noncontr-eq}), we obtain that $\mathbb{E}_{\tau^m}\,\{v^\infty_i(s^\ell)|s^n=s\}$ is non-decreasing for $\ell\in\{n,n+1,\ldots,m\}$. Hence,
\begin{eqnarray*}
u_i(\tau^m|s) &=& \sum_{z\in S^\infty} \mathbb{P}_{\tau^m}\,\{s^m=z\,|\,s^n=s\}\cdot u_i(\sigma_i^\infty,\tau_{-i}|z)\\ 
&\geq& \sum_{z\in S^\infty} \mathbb{P}_{\tau^m}\,\{s^m=z\,|\,s^n=s\}\cdot v^\infty_i(z)\\
 &=& \mathbb{E}_{\tau^m}\,\{v^\infty_i(s^m)\,|\,s^n=s\} \,\geq\, v^\infty_i(s).
\end{eqnarray*}
As the strategy profile $\tau^m$ converges to $\tau$ as $m$ tends to infinity from state $s$, by Lemma~\ref{cont-exp} we obtain $u_i(\tau|s)\geq v^\infty_i(s)$, as desired.
\qed\end{proof}

\subsection{The Secure Equilibrium}

Now we argue that thanks to the properties of this new game $G^\infty$, we are able to identify an interesting strategy profile in this game, which can be enhanced with punishment strategies in order to become a secure equilibrium in the original game $G$.
	
Recall that $\sigma_{-i}$ and respectively $\sigma_{-i}^\infty$ are strategy profiles for player $i$'s opponents in $G$ and respectively in $G^\infty$, such that they induce an optimal strategy profile in $G_i(s)$ and respectively $G_i^\infty(s')$, for every state $s\in S$ and every state $s' \in S^\infty$. For each player $j\in N\setminus\{i\}$, let $\sigma_{-i,j}$ and respectively $\sigma^\infty_{-i,j}$ be player $j$'s strategy in these strategy profiles. These strategies will play the role of punishing player $i$ if player $i$ deviates. 

Let $\rho^\ast$ be a strategy profile in $G^\infty$ which minimizes the sum of the expected payoffs of all the players, i.e. for any strategy profile $\rho$ in $G^\infty$ we have $\sum_{i\in N}u_i(\rho^*)\leq \sum_{i\in N}u_i(\rho)$. Such a strategy profile exists due to the compactness of the set of strategy profiles in $G^\infty$ and due to Lemma~\ref{cont-exp}. The idea is to define a secure equilibrium in $G$ in the following way: follow $\rho^\ast$, unless a deviation occurs. If player $i$ deviates, then his opponents should punish player $i$ with $\sigma^\infty_{-i}$ as long as he chooses actions in $G^\infty$, and they should punish him with $\sigma_{-i}$ as soon as he chooses an action out of $G^\infty$. Let us specify this strategy profile more precisely.

We first define a function $L$, which assigns to each state the player who has to be punished from this state, or $\bot$ if nobody has to be punished. The idea here is to remember the first player who deviated from the strategy profile $\rho^\ast$. For the initial state $\widetilde{s}$, we have $L(\widetilde{s})=\bot$. For other states, we define it by induction. Suppose that we have defined $L(s)$ for some state $s\in S$. Then, for a state $s'\in S(s,a)$, where $a\in A(s)$, we set:
  \begin{equation*}
    L(s') :=
    \begin{cases}
      \bot & \mbox{ if $L(s)= \bot$ and $\rho^*_{i(s)}(s)=a$,}\\

      i(s) & \mbox{ if $L(s)= \bot$ and $\rho^*_{i(s)}(s)\in A(s)\setminus\{a\}$,}\\

      L(s) & \mbox{ otherwise (i.e. when $L(s) \not=\bot$).}
    \end{cases}
  \end{equation*}
Now we define a strategy $\tau_j$ for every player $j\in N$ as follows: for every state $s\in S_j$ let 
  \begin{equation*}
    \tau_j(s) :=
    \begin{cases}
      \rho^*_j(s) & \mbox{ if $L(s)=\bot$,}\\

      \text{arbitrary action in $A^\infty(s)$} & \mbox{ if $L(s)=j$ and $s\in S^\infty$,}\\
			
			\text{arbitrary action in $A(s)$} & \mbox{ if $L(s)=j$ and $s\in S\setminus S^\infty$,}\\

      \sigma_{-i,j}^\infty(s) & \mbox{ if $L(s)=i\neq j$ and $s\in S^\infty$,}\\

      \sigma_{-i,j}(s) & \mbox{ if $L(s)=i\neq j$ and $s\in S\setminus S^\infty$.}
    \end{cases}
  \end{equation*}
We will show that $\tau=(\tau_j)_{j\in N}$ is a secure equilibrium in $G$. We remark that in the second case, when $L(s)=j$ and $s\in S^\infty$, it is not necessary to have any restriction on the action $\tau_j(s)$, and we only require $\tau_j(s)\in A^\infty(s)$ because it simplifies the arguments of the proof. 

Consider a strategy $\tau_i'$ for player~$i$.
The following lemma says that, in any state $s\in S^\infty$, it is strictly worse for player $i$ to deviate to an action outside $A^\infty(s)$, unless another player deviated before him (that case is irrelevant for our goal to show that $\tau$ is a secure equilibrium).

\begin{lemma}\label{claim1}
For each state $s\in S^\infty_i$ such that $L(s)\in\{\bot,i\}$ and $\tau'_i(s)\in A(s)\setminus A^\infty(s)$, we have $u_i(\tau'_i,\tau_{-i}|s)<u_i(\tau|s)$. 
\end{lemma}

\begin{proof}
Suppose the opposite, and let $a=\tau'_i(s)$. Then, there is a $k\in\mathbb{N}$ such that $a\in A^k(s)\setminus A^{k+1}(s)$. In particular, $v^{k}_i(s)>v^k_i(s,a)$. Moreover, because $L(z)=i$ and $z\in S\setminus S^\infty$ for every state $z\in S(s,a)$, the strategy profile $\tau_{-i}$ tells player $i$'s opponents to follow $\sigma_{-i}$ from every state $z\in S(s,a)$. Hence, by Lemmas~\ref{inc_vk} and~\ref{goodlimit}
\begin{eqnarray*}
u_i(\tau'_i,\tau_{-i}|s) &=& \sum_{z\in S(s,a)}q(s,a)(z)\cdot u_i(\tau'_i,\sigma_{-i}|z)\,\leq\,\sum_{z\in S(s,a)}q(s,a)(z)\cdot v_i(z)\\
&= & v_i(s,a) \,\leq\, v^k_i(s,a) \,<\, v^k_i(s) \,\leq\,v^\infty_i(s) .
\end{eqnarray*}

Notice that $v^\infty_i(s)\,\leq\,u_i(\tau|s)$. Indeed, if $L(s)=\bot$ then $u_i(\tau|s)=u_i(\rho^*|s)\geq v^\infty_i(s)$ by Lemma~\ref{inf-game}. Now, if $L(s)=i$, then when playing according to $\tau$ the following properties hold in the subgame $G(s)$: only states in $S^\infty$ are visited, player $i$ plays the action $\tau_i(s')\in A^\infty(s')$ in all states $s'\in S^\infty_i$ that are reached after $s$, and player $i$'s opponents follow $\sigma^\infty_{-i}$. Hence, in the subgame $G(s)$, the strategy profile $\tau$ does not leave $G^\infty$ when it is played, so by Lemma~\ref{inf-game}, we obtain again that $u_i(\tau|s)\geq v^\infty_i(s)$. 

By combining the above, we find $u_i(\tau'_i,\tau_{-i}|s)<v^\infty_i(s)\,\leq\,u_i(\tau|s)$, as desired.
\qed\end{proof}

With the help of the above lemma, we can handle deviations to actions outside $G^\infty$. The next lemma is about deviations to actions inside $G^\infty$. It claims that player $i$ does not get a better payoff if he deviates to an action inside the game $G^\infty$, given no deviation has occurred before.

\begin{lemma}\label{claim2}
For each state $s\in S^\infty_i$ such that $L(s)=\bot$ and $\tau'_i(s)\in A^\infty(s)\setminus\{\rho^*_{i}(s)\}$, we have $u_i(\tau'_i,\tau_{-i}|s)\leq u_i(\tau|s)$. 
\end{lemma}

\begin{proof}

Suppose that state $s\in S^\infty_i$ is such that $L(s)=\bot$ and $\tau'_i(s)\in A^\infty(s)\setminus\{\rho^*_{i}(s)\}$.

Let $Z$ denote the set of states $z\in S^\infty_i$ such that (1) $L(z)\in\{\bot,i\}$ and $\tau'_i(z)\in A(z)\setminus A^\infty(z)$, and (2) $z$ is minimal with property (1) in the sense that there is no state $z'\neq z$ on the unique path from the initial state $\widetilde{s}$ to $z$ with property (1). Let $\mu_i$ be the strategy for player $i$ as follows: $\mu_i$ coincides with $\tau'_i$ as long as no state in $Z$ is reached, and as soon as a state $z\in Z$ is reached, $\mu_i$ follows $\tau_i$ in the subgame $G(z)$. In view of Lemma~\ref{claim1}, we have $u_i(\tau'_i,\tau_{-i}|s)\leq u_i(\mu_i,\tau_{-i}|s)$. Therefore, it suffices to show that $u_i(\mu_i,\tau_{-i}|s)\leq u_i(\tau|s)$.

Since $L(s)=\bot$, no state in $Z$ is visited on the unique path from the initial state $\widetilde{s}$ to $s$. Hence, $\tau'_i(s)=\mu_i(s)$. Let $a$ denote this action. Then, for any state $w\in S(s,a)$, we have $L(w)=i$ and $w\in S^\infty$. So, when playing according to $(\mu_i,\tau_{-i})$ the following properties hold in the subgame $G(s)$: only states in $S^\infty$ are visited, player $i$'s action $\mu_i(s')$ belongs to $A^\infty(s')$ in all states $s'\in S^\infty_i$ that are reached after $s$, and player $i$'s opponents follow $\sigma^\infty_{-i}$. Hence, in the subgame $G(s)$, the strategy profile $(\mu_i,\tau_{-i})$ does not leave $G^\infty$ when it is played, so by the definition of $\sigma^\infty_{-i}$, we obtain that $u_i(\mu_i,\tau_{-i}|s)\leq v^\infty_i(s)$. 

Since $s\in S^\infty$, we have $u_i(\tau|s)\geq v^\infty_i(s)$ by Lemma~\ref{inf-game}. Therefore, $u_i(\mu_i,\tau_{-i}|s)\leq v^\infty_i(s)\leq u_i(\tau|s)$, as desired.
\qed\end{proof}

Now we prove that $\tau$ is a secure equilibrium. It follows easily from Lemmas~\ref{claim1} and~\ref{claim2} that $u_i(\tau'_i,\tau_{-i})\leq u_i(\tau)$. This means that $\tau$ is a Nash equilibrium.

Furthermore, consider the case where $u_i(\tau'_i,\tau_{-i})= u_i(\tau)$. Then, Lemma~\ref{claim1} implies that it has probability zero that a state $s\in S^\infty_i$ with $\tau'_i(s)\in A(s)\setminus A^\infty(s)$ is reached under $(\tau'_i,\tau_{-i})$. Consequently, when playing according to $(\tau'_i,\tau_{-i})$ the following properties hold in the game $G$: (1) only states in $S^\infty$ are visited, and (2) in all states $s\in S^\infty_i$ that are reached, player $i$'s action $\tau'_i(s)$ belongs to $A^\infty(s)$, and (3) in all states $s\in S^\infty_j$, where $j\in N\setminus\{i\}$, that are reached, player $j$ plays the action given by $\rho^*_{j}(s)$ if $L(s)=\bot$ and by $\sigma^\infty_{-i,j}(s)$ if $L(s)=i$. Hence, the strategy profile $(\tau'_i,\tau_{-i})$ does not leave $G^\infty$ when it is played. As $u_i(\tau'_i,\tau_{-i})= u_i(\tau)$, the definition of $\rho^*$ yields 
$$\sum_{j\in N\setminus\{i\}}u_j(\tau'_i,\tau_{-i})\geq \sum_{j\in N\setminus\{i\}}u_j(\rho^*)=\sum_{j\in N\setminus\{i\}}u_j(\tau).$$
Thus, $\tau$ is indeed a secure equilibrium in the game $G$, and the proof of Theorem~\ref{main-th1} is complete.

\section{The Proof of Theorem \ref{main-th2}}\label{sec:proof2}

In this section, we provide a proof for Theorem~\ref{main-th2}. Consider a perfect information game $G$ with deterministic transitions, with arbitrary state and action spaces. Assume that the payoff functions are Borel measurable and have a finite range $M$, i.e. every player $i$'s payoff function $u_i$ is only taking values in $M$. Assume also that $M$ contains at least two elements, otherwise the game is trivial. 

Let $R=\max_{m\in M}|m|$ and $d=\min_{m,m'\in M,\,m\neq m'}|m-m'|$, and then choose $\delta=\tfrac{d}{2|N|R}$. We denote by $G^\delta$ the game $G$ with a new payoff function $u_i^\delta$ for every  player $i\in N$, defined  as follows: for every play $p\in\mathcal{P}$, let 
$$u^\delta_i(p)=u_i(p)-\delta\cdot\sum_{j\in N,\,j\neq i}u_j(p).$$ Notice that, for two plays $p,p'\in\mathcal{P}$, if we have $u_i(p)< u_i(p')$ then 
$$u_i^\delta(p') - u_i^\delta(p) = (u_i(p')-u_i(p))-\delta\cdot\sum_{j\in N,\,j\neq i}(u_j(p')-u_j(p)) \geq d-\delta\cdot (|N|-1)\cdot 2R > 0,$$
so $u_i^\delta(p)< u_i^\delta(p')$ holds too. Consequently, $u_i^\delta(p)\geq u_i^\delta(p')$ implies $u_i(p)\geq u_i(p')$.

Now suppose that $\sigma^\ast$ is a Nash equilibrium in $G^\delta$. Then, due to the previous observation, $\sigma^\ast$ is also a Nash equilibrium in the original game $G$. Now we show that $\sigma^\ast$ is a secure equilibrium in $G$. So, suppose that $\tau_i$ is a strategy for some player $i\in N$ such that $u_i(\tau_i,\sigma^\ast_{-i})=u_i(\sigma^\ast)$. Since $\sigma^\ast$ is a Nash equilibrium in $G^\delta$, we also have $u_i^\delta(\tau_i,\sigma^\ast_{-i})\leq u_i^\delta(\sigma^\ast)$. Hence $$\sum_{j\in N,\,j\neq i}u_j(\tau_i,\sigma^\ast_{-i})\geq \sum_{j\in N,\,j\neq i}u_j(\sigma^\ast),$$ which proves that $\sigma^\ast$ is a secure equilibrium in $G$ indeed.

It remains to prove that $G^\delta$ admits a Nash equilibrium. We only provide a sketch, since similar constructions are well known (cf. the result of Mertens and Neyman in \cite{Mer87}, and also \cite{TR97}), and many of these ideas also appeared in the proof of Theorem~\ref{main-th1}. The important property of $G^\delta$ is that the payoff functions $u^\delta_i$, for $i\in N$, are Borel measurable and have a finite range. By applying a corollary of Martin \cite{Mar75}, for any $i \in N$, player $i$ has a subgame-perfect optimal strategy $\sigma_i$ and his opponents have a subgame-perfect optimal strategy profile $\sigma_{-i}$ in the zero-sum game $G_i$, in which player $i$ maximizes $u^\delta_i$ and the other players are jointly minimizing $u^\delta_i$. It can be checked easily that the following strategy profile is a Nash equilibrium in $G^\delta$: every player $i$ should use the strategy $\sigma_i$. As soon as a player deviates, say player $i$ plays another action, then the other players should punish player $i$ in the remaining game by switching to the strategy profile $\sigma_{-i}$.

\section{Concluding Remarks}\label{sec:ccl}

\noindent\textbf{Lexicographic objectives:} In the proof of Theorem~\ref{main-th1}, we in fact showed the existence of a sum-secure equilibrium (see Remark~\ref{rem-sum-secure}). A very similar proof can be given to show that, if $F: \mathbb{R}^{|N|} \rightarrow \mathbb{R}$ is a continuous and bounded function, then there exists a Nash equilibrium $\sigma^\ast$ such that the following property holds for every player $i\in N$: if $\tau_i$ is a strategy for player $i$ such that $u_i(\tau_i,\sigma^\ast_{-i})=u_i(\sigma^\ast)$ then $F(u_j(\sigma^\ast)_{j\in N})\leq F(u_j(\sigma_i,\sigma^\ast_{-i})_{j\in N})$. In fact, this is closely related to lexicographic preferences: each player's first objective is to maximize his payoff, but in case of a tie between strategies, the secondary objective is to minimize the function $F$.\smallskip


\noindent\textbf{Subgame-perfect secure equilibrium:} Brihaye et al \cite{BBDG13} introduced the concept of subgame-perfect secure equilibrium, and showed its existence in two-player quantitative reachability games. We do not know if Theorem~\ref{main-th1} can be extended to subgame-perfect secure equilibrium. Perhaps it is possible to make use of the recently developed techniques for subgame-perfect equilibria in \cite{FKMSSV10,PS11}. 

However, Theorem~\ref{main-th2} cannot be extended to subgame-perfect secure equilibrium, as even a subgame-perfect equilibrium does not always exist in perfect information games with deterministic transitions and finitely many payoffs, as is shown by an example in~\cite{SV03}.\smallskip

\noindent\textbf{Algorithmic result:} Let us consider multi-player quantitative reachability games played on a finite graph, with deterministic transitions, where each payoff is determined by the number of moves it takes to get in a particular set of states. As a corollary of Theorem~\ref{main-th1} and some results of \cite{BBDG13} (the proof of Theorem~4.1, Proposition~4.5 and Remark~4.7), we derive an algorithm to obtain, in {\textsc{ExpSpace}}, a secure equilibrium such that finite payoffs are bounded by $2 \cdot |N| \cdot |S|$ in such games. We intend to further investigate algorithmic questions for other classes of objectives.\medskip

\noindent\textbf{Acknowledgment:} We would like to thank Dario Bauso, Thomas Brihaye, V\'{e}ronique Bruy\`{e}re and Guillaume Vigeral for valuable discussions on the concept of secure equilibrium.

\bibliographystyle{abbrv}
\bibliography{secure}

\end{document}